\documentclass[11pt]{article}
\usepackage[paperwidth=560pt,paperheight=700pt,top=72pt,right=67pt,bottom=72pt,left=72pt]{geometry}
\usepackage[utf8]{inputenc}
\usepackage{amsfonts}
\usepackage{amssymb}
\usepackage{amsthm}
\usepackage{amsmath}
\usepackage{graphics}
\usepackage{graphicx}
\usepackage{enumitem}
\usepackage[T1]{fontenc}
\usepackage{epstopdf}
\usepackage{times}
\usepackage{cite}
\DeclareUnicodeCharacter{2217}{*}
\newtheorem{thm}{Theorem}

\newtheorem{corollary}[thm]{Corollary}
\date{}

\title{\textbf{Study on $(r,s)$- Generalised Transformation Graphs,\\ A Novel Perspective Based on Transformation Graphs}}
\author{\begin{tabular}{c}
		Parvez Ali$^{1}$, Annmaria Baby$^{2,*}$, D. Antony Xavier$^{2}$, Theertha Nair A$^2$, \\Haidar Ali$^{3}$, Syed Ajaz K. Kirmani$^4$\\ \\
	\end{tabular}\\
	\begin{tabular}{c}
		$^1$ Department of Mechanical Engineering, College of Engineering,\\ Qassim University, Buraydah, 51452, Saudi Arabia.\\
			$^2$Department of Mathematics, Loyola College, University of Madras, Chennai, India.\\
			$^3$ University Community College, Government College University, Faisalabad-Pakistan.\\
			$^4$ Department of Electrical Engineering, College of Engineering,\\ Qassim University, Buraydah, 51452, Saudi Arabia.\\\\
		Mail: p.ali@qu.edu.sa, annbaby179@gmail.com, dantonyxavierlc@gmail.com,\\ theerthanaira@gmail.com, haidar3830@gmail.com, s.kirmani@qu.edu.pk
	\end{tabular}
}
\begin{document}
	\maketitle
	
	\begin{abstract}
		For a graph $\mathbb{Q}=(\mathbb{V},\mathbb{E})$, the transformation graphs are defined as graphs with vertex set being $\mathbb{V(Q)} \cup \mathbb{E(Q)}$ and edge set is described following certain conditions. In comparison to the structure descriptor of the original graph $\mathbb{Q}$, the topological descriptor of its transformation graphs displays distinct characteristics related to structure. Thus, a compound's transformation graphs descriptors can be used to model a variety of structural features of the underlying chemical. In this work, the concept of transformation graphs are extended giving rise to novel class of graphs, the $(r,s)$- generalised transformation graphs, whose vertex set is union of $r$ copies of $\mathbb{V(Q)}$ and $s$ copies of $\mathbb{E(Q)}$, where, $r, s \in N$ and the edge set are defined under certain conditions. Further, these class of graphs are analysed with the help of first Zagreb index. Mainly, there are eight transformation graphs based on the criteria for edge set, but under the concept of $(r,s)$- generalised transformation graphs, infinite number of graphs can be described and analysed.
	\end{abstract}
	
	{\bf Keywords}: Transformation graphs; Generalised transformations; Topological descriptor; First Zagreb index.\\\\
	{\bf{2020 Mathematics Subject Classification:}} 05C09, 05C90, 05C92, 92E10
	
	\section{Introduction} 
	
	In theoretical chemistry, chemical compounds are visualized as molecular graphs with edges acting in for chemical bonds and vertices for atoms. This concept, known as chemical graph theory provides a link between chemistry and mathematics, redirecting to quantitative study of chemical compounds through topological descriptors. Topological descriptors are quantitative attributes of a graph that remain invariant under graph isomorphism. It helps in a highly appealing way to identify the physical, chemical, biological or pharmacological aspects of a chemical structure. It has broad applications in many different domains, including biology, informatics, and chemistry \cite{wiener,balaban,gutmanM1,randic1,randic2,TIapps,symmetry,LLL}. Interest has been sparked by  the use of topological descriptors in quantitative analyses of the relationship between structure-activity and structure-property in chemistry \cite{app1,app2}. The vertex degree is the base for many of these descriptors \cite{degree}. In 1947, H. Wiener developed the Wiener index, the first known distance based index \cite{wiener}. Later, in 1972, the Zagreb degree-based topological descriptors were introduced by Gutman and Trinajsti\'{c} \cite{gutmanM1}. Much attention has been paid to degree-based topological descriptors, particularly Zagreb indices. For further details, see \cite{rect} and the other references cited within. Their mathematical features are detailed in the survey \cite{surv}.\\	\\
	Consider a graph $\mathbb{Q(V,E)}$ with $|\mathbb{V(Q)}|=\mathfrak{n}$ and $|\mathbb{E(Q)}|=\mathfrak{m}$. For a vertex $\alpha \in \mathbb{V(Q)}$, the degree of the vertex $\alpha$ in $\mathbb{Q}$ is denoted as $d_{\mathbb{Q}}(\alpha)$, which is the count of edges incident to that vertex. Two vertices are said to be adjacent, if they are connected by an edge and two edges are said to be adjacent, if they have a common end vertex. The first Zagreb index \cite{gutmanM1} and the Forgotten index (F-index) \cite{forgot} are defined as follows respectively.
	\begin{align}
		\mathbb{M}_1\mathbb{(Q)}=\sum_{\alpha \in \mathbb{V(Q)}}\big(d_{\mathbb{Q}}(\alpha)\big)^2.
	\end{align}
	\begin{align}
		\mathbb{F}\mathbb{(Q)}= \displaystyle \sum_{\alpha \in \mathbb{V(Q)}} \big(d_{\mathbb{Q}}(\alpha)\big)^3.
	\end{align} 
	
	The foundation stone for transformation graphs was laid back in 1966 by Mehdi Behzad \cite{total66,total67} with introducing the idea of total graphs.  Later, in 1973, Sampathkumar and Chikkodimath \cite{semi} came up with the idea of semi total-point graph and semi total-line graph. Thus the concept of total transformation graphs was developed.  In \cite{2001}, some new graphical transformations generalising the concept of total graph was developed by resarchers. There are eight different total transformation graphs depending on the  adjacency of each vertices in the transformation graph.  The transformation graph's topological descriptor must reflect different structural features than that of the basic molecular graph. With the help of same class of descriptors, a wide range of structural properties of the underlying molecules could be modelled. See \cite{c1,c2,c3,c4, L1,L2,L3} for applications of such concepts.\\
	
	A transformation graph, $\mathbb{Q}^{uvw}$ of a graph $\mathbb{Q}$ is developed in such a way that their vertex set, $\mathbb{V}(\mathbb{Q}^{uvw})$ is $\mathbb{V(Q)} \cup \mathbb{E(Q)}$ and their edge set is defined under certain criterias which follows.\\
	For three variables $u,v,w$, the vertices $\alpha, \beta \in \mathbb{V}(\mathbb{Q}^{uvw})$ are adjacent if and only if
	\begin{enumerate}[label=(\roman*)]
		\item $\alpha, \beta \in \mathbb{V(Q)}$, $\alpha, \beta$ are adjacent in $\mathbb{Q}$, if $u=+$ and $\alpha, \beta$ are not adjacent in $\mathbb{Q}$, if $u=-$.
		\item $\alpha, \beta \in \mathbb{E(Q)}$, $\alpha, \beta$ are adjacent in $\mathbb{Q}$, if $v=+$ and $\alpha, \beta$ are not adjacent in $\mathbb{Q}$, if $v=-$.
		\item $\alpha \in \mathbb{V(Q)}$ and $ \beta \in \mathbb{E(Q)}$, $\alpha, \beta$ are incident in $\mathbb{Q}$, if $w=+$ and $\alpha, \beta$ are not incident in $\mathbb{Q}$, if $w=-$.
	\end{enumerate}
	 Eight transformation graphs can be developed for a given graph $\mathbb{Q}$, based on the distinct 3-permutations of $\{+,-\}$, Additionally, it can be noted that, there are four pairs of transformation graphs among the eight graphs, which are mutually complementary, that is one is isomorphic to the complement of its pair. A sample graph, $\mathbb{Q}$ along with its transformation graphs are presented in Figure \ref{fig:trans8}.
	\begin{figure}[ht!]
		\centering
		\includegraphics[width=1.05\linewidth]{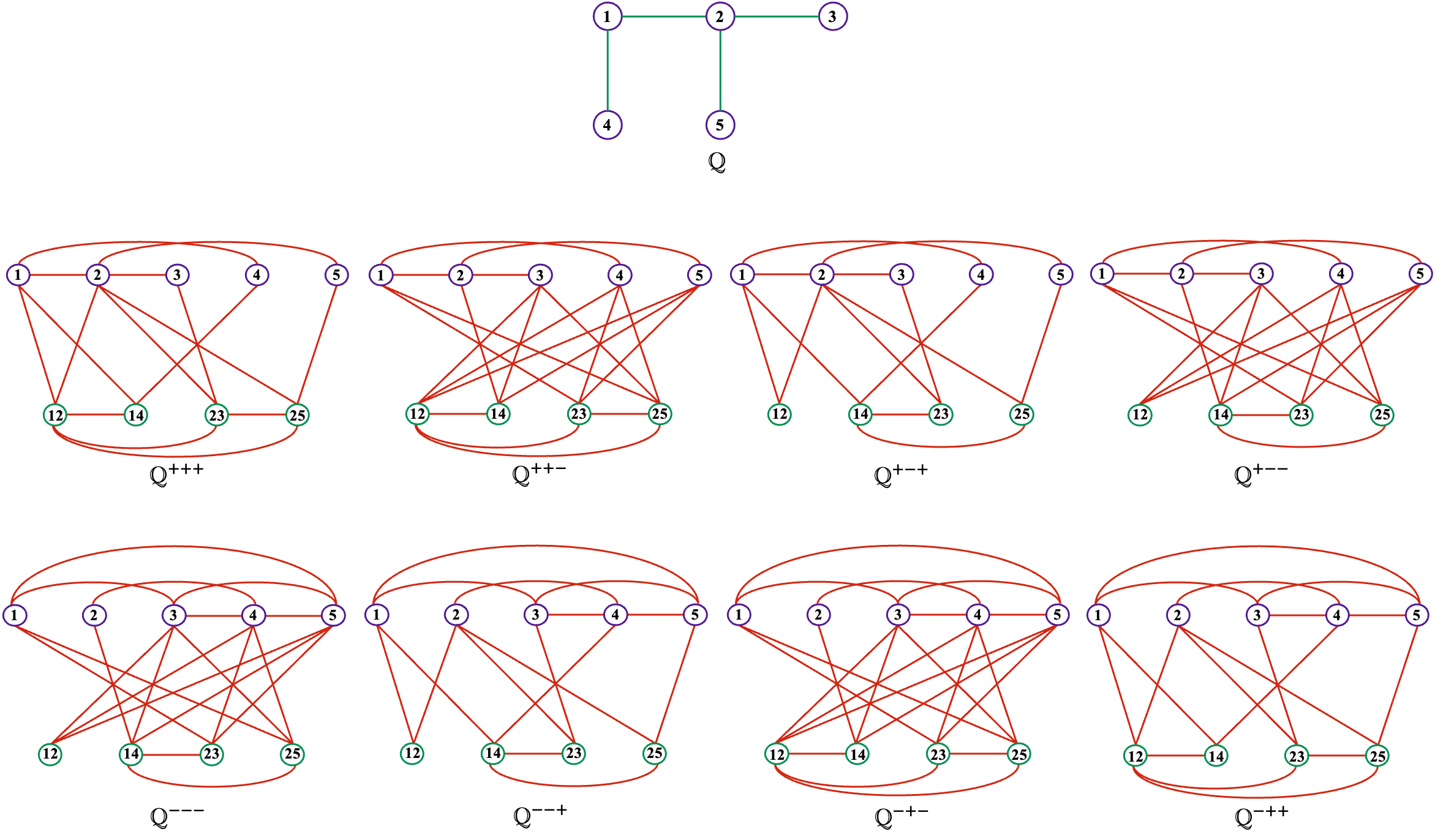}
		\caption{A sample graph $\mathbb{Q}$ and its eight transformation graphs.}
		\label{fig:trans8}
	\end{figure}
	
	Under these conditions, the transformation graphs are restricted to a count of eight. Motivated by this concept, an attempt has been initiated to describe infinite number of transformations from a given graph $\mathbb{Q}$. This novel perspective lead to the concept of $(r,s)$- generalised transformation graphs. Experts in this field have introduced and studied the concept of $k$- generalised transformation graphs based on transformation graphs \cite{Kth1,Kth2}.
	
	\section{$(r,s)$- Generalised Transformation Graphs}
	Consider the simple connected graph $\mathbb{Q}=\big(\mathbb{V}, \mathbb{E}\big)$. Take $r$ copies of vertex set $\mathbb{V(Q)}$, say $\mathbb{V}_1\mathbb{(Q)}$, $\mathbb{V}_2\mathbb{(Q)}$, ...,$ \mathbb{V}_r\mathbb{(Q)}$ and $s$ copies of edge set $\mathbb{E(Q)}$, say $\mathbb{E}_1\mathbb{(Q)}, \mathbb{E}_2\mathbb{(Q)}, ..., \mathbb{E}_s\mathbb{(Q)}$, where $r,s \in N$. \\
	Then, the $(r,s)$- generalised transformation graph, $\mathbb{Q}_{rs}^{xyz}$ with\\
	$x=(x_1, x_2, ..., x_r); x_g \in \{+,-\}$,\\
	$y=(y_1, y_2, ..., y_s); y_h \in \{+,-\}$ and \\ 
	$z=(z_{11}, z_{12}, ..., z_{1s}, z_{21}, z_{22}, ..., z_{2s}, ..., z_{r1}, z_{r2}, ..., z_{rs}); z_{gh} \in \{+,-\}$\\
	where $1 \leq g \leq r; 1 \leq h \leq s$, has vertex set $ \bigg(\displaystyle\bigcup_{g=1}^{r} \mathbb{V}_g \bigg)\bigcup \bigg(\displaystyle\bigcup_{h=1}^{s} \mathbb{E}_h \bigg)$ and edge set is defined such that $\alpha, \beta \in \mathbb{V}(\mathbb{Q}_{rs}^{uvw})$ are adjacent if and only if
	\begin{enumerate}[label=(\roman*)]
		\item 
		$\alpha, \beta \in \mathbb{V}_g\mathbb{(Q)}$, $\alpha, \beta$ are adjacent in $\mathbb{Q}$, if $x_g=+$ and $\alpha, \beta$ are not adjacent in $\mathbb{Q}$, if $x_g=-$. 
		\item
		$\alpha, \beta \in \mathbb{E}_h\mathbb{(Q)}$, $\alpha, \beta$ are adjacent in $\mathbb{Q}$, if $y_h=+$ and $\alpha, \beta$ are not adjacent in $\mathbb{Q}$, if $y_h=-$. 
		\item
		$\alpha \in \mathbb{V}_g\mathbb{(Q)}$ and $\beta \in \mathbb{E}_h\mathbb{(Q)}$, $\alpha$, $\beta$ are incident in $\mathbb{Q}$, if $z_{gh}=+$ and $\alpha, \beta$ are not incident in $\mathbb{Q}$, if $z_{gh}=-$. 
	\end{enumerate}
	
	\subsection*{Special Notations:}
	For the $(r,s)$- generalised transformation graph, $\mathbb{Q}_{rs}^{xyz}$, consider the sets\\
	$x=(x_1, x_2, ..., x_r), x_g \in \{+,-\}$,\\ $y=(y_1, y_2, ..., y_s), y_h \in \{+,-\}$,\\ 
	$z=(z_{11}, z_{12}, ..., z_{1s}, z_{21}, z_{22}, ..., z_{2s}, ..., z_{r1}, z_{r2}, ..., z_{rs}), z_{gh} \in \{+,-\}$\\
	where $1 \leq g \leq r; 1 \leq h \leq s$. 
	\begin{enumerate}[label=(\roman*)]
		\item If $x_g=+$, for $1 \leq g \leq r$, then the graph is denoted as $\mathbb{Q}_{rs}^{+yz}$. 
		\item	If $x_g=-$, for $1 \leq g \leq r$, then the graph is denoted as $\mathbb{Q}_{rs}^{-yz}$.
		\item If $y_h=+$, for $1 \leq h \leq s$, then the graph is denoted as $\mathbb{Q}_{rs}^{x+z}$.
		\item	If $y_h=-$, for $1 \leq h \leq s$, then the graph is denoted as $\mathbb{Q}_{rs}^{x-z}$.
		\item	If $z_{gh}=+$, for $1 \leq g \leq r; 1 \leq h \leq s$, then the graph is denoted as $\mathbb{Q}_{rs}^{xy+}$.
		\item	If $z_{gh}=-$, for $1 \leq g \leq r; 1 \leq h \leq s$, then the graph is denoted as $\mathbb{Q}_{rs}^{xy-}$.
		\item	If $x$, $y$ and $z$ are such that $x_g=+, y_h=+, z_{gh}=+$, for $1 \leq g \leq r; 1 \leq h \leq s$, then the graph is denoted as $\mathbb{Q}_{rs}^{+++}$. 
		\item	If $x$, $y$ and $z$ are such that $x_g=-, y_h=-, z_{gh}=-$, for $1 \leq g \leq r; 1 \leq h \leq s$, then the graph is denoted as $\mathbb{Q}_{rs}^{---}$.
	\end{enumerate}
	
	\subsection*{Illustration:}
	Consider a sample graph $\mathbb{Q(V,E)}$ with $\mathbb{V(Q)}=\{1,2,3,4,5\}$ and $\mathbb{E(Q)}=\{12,14,23,25\}$ as given in Figure \ref{fig:1}.
	\begin{figure}[ht!]
		\centering
		\includegraphics[width=0.24\linewidth]{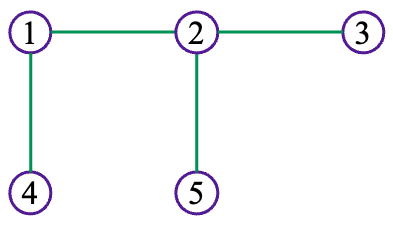}
		\caption{A sample graph, $\mathbb{Q}$.}
		\label{fig:1}
	\end{figure}
	\begin{itemize}
		\item {Example 1:}\\
		Choose $r=3$ and $s=3$. Hence, we have 3 copies each of $\mathbb{V(Q)}$ and $\mathbb{E(Q)}$. \\
		Also, take
		$x=\{x_1, x_2, x_3\}$,
		$y=\{y_1, y_2, y_3\}$ and
		$z=\{z_{11}, z_{12}, z_{13}, z_{21}, z_{22}, z_{23}, z_{31}, z_{32}, z_{33}\}$,\\
		such that $x_g=+, y_h=+, z_{gh}=+; \text{for } g,h \in  \{1, 2, 3\}$. \\
		The resulting transformation graph is $\mathbb{Q}_{33}^{+++}$, as demonstrated in Figure \ref{fig:Q32plus}.
		\begin{figure}[ht!]
			\centering
			\includegraphics[width=1.05\linewidth]{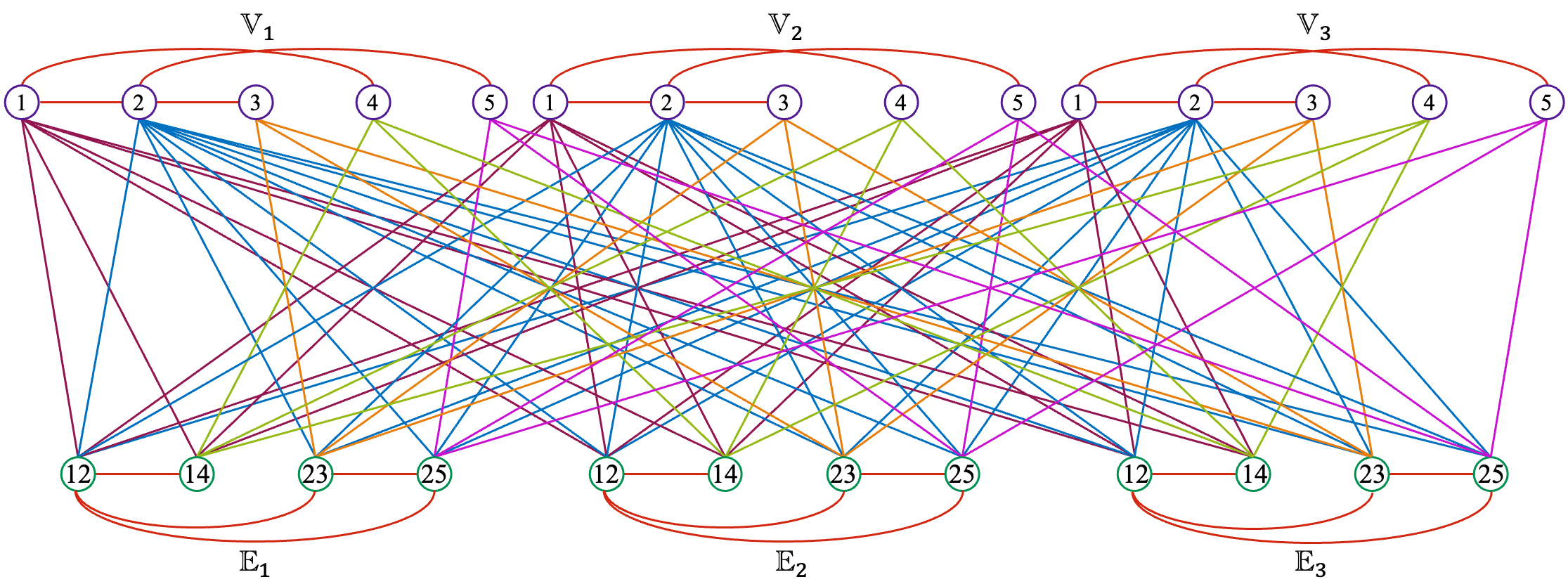}
			\caption{The transformation graph, $\mathbb{Q}_{33}^{+++}$.}
			\label{fig:Q32plus}
		\end{figure}
		
		\item {Example 2}\\
		Choose $r=2$ and $s=1$. \\
		Take,
		$x=\{x_1, x_2\}=\{+,-\}$,
		$y=\{y_1\}=\{-\}$ and
		$z=\{z_{11}, z_{21}\}=\{+,+\}$.\\
		The resulting transformation graph is $\mathbb{Q}_{21}^{x-+}$, as demonstrated in Figure \ref{fig:Q21eg}.\\
		\begin{figure}[ht!]
			\centering
			\includegraphics[width=0.75\linewidth]{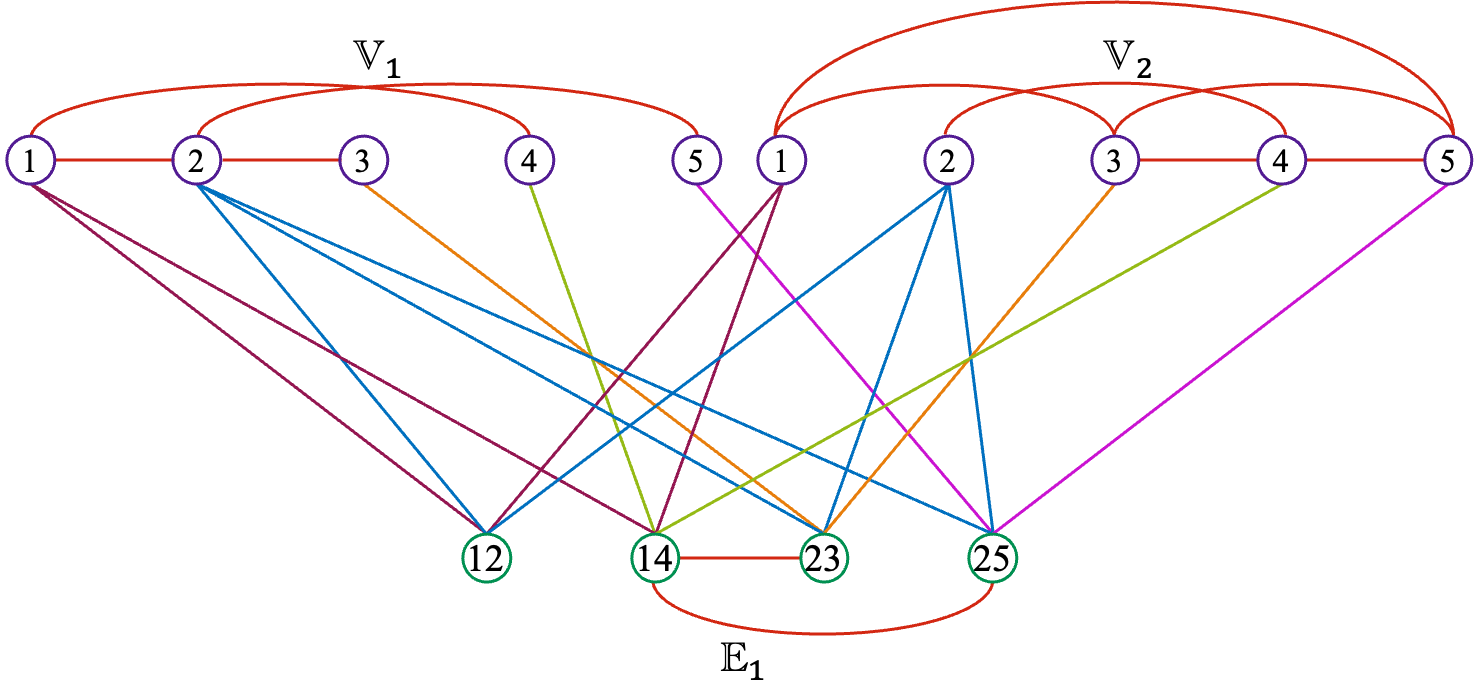}
			\caption{The transformation graph, $\mathbb{Q}_{21}^{x-+}$, where $x=\{+,-\}$.}
			\label{fig:Q21eg}
		\end{figure}
	\newpage	
		\item {Example 3}\\
		Choose $r=2$ and $s=2$. \\
		Take,
		$x=\{x_1, x_2\}=\{-,-\}$,
		$y=\{y_1,y_2\}=\{-,+\}$ and\\
		$z=\{z_{11}, z_{12}, z_{21}, z_{22}\}=\{+,-,+,-\}$.\\
		The resulting transformation graph is $\mathbb{Q}_{22}^{-yz}$, as demonstrated in Figure \ref{fig:Q22eg}.
		\begin{figure}[ht!]
			\centering
			\includegraphics[width=0.7\linewidth]{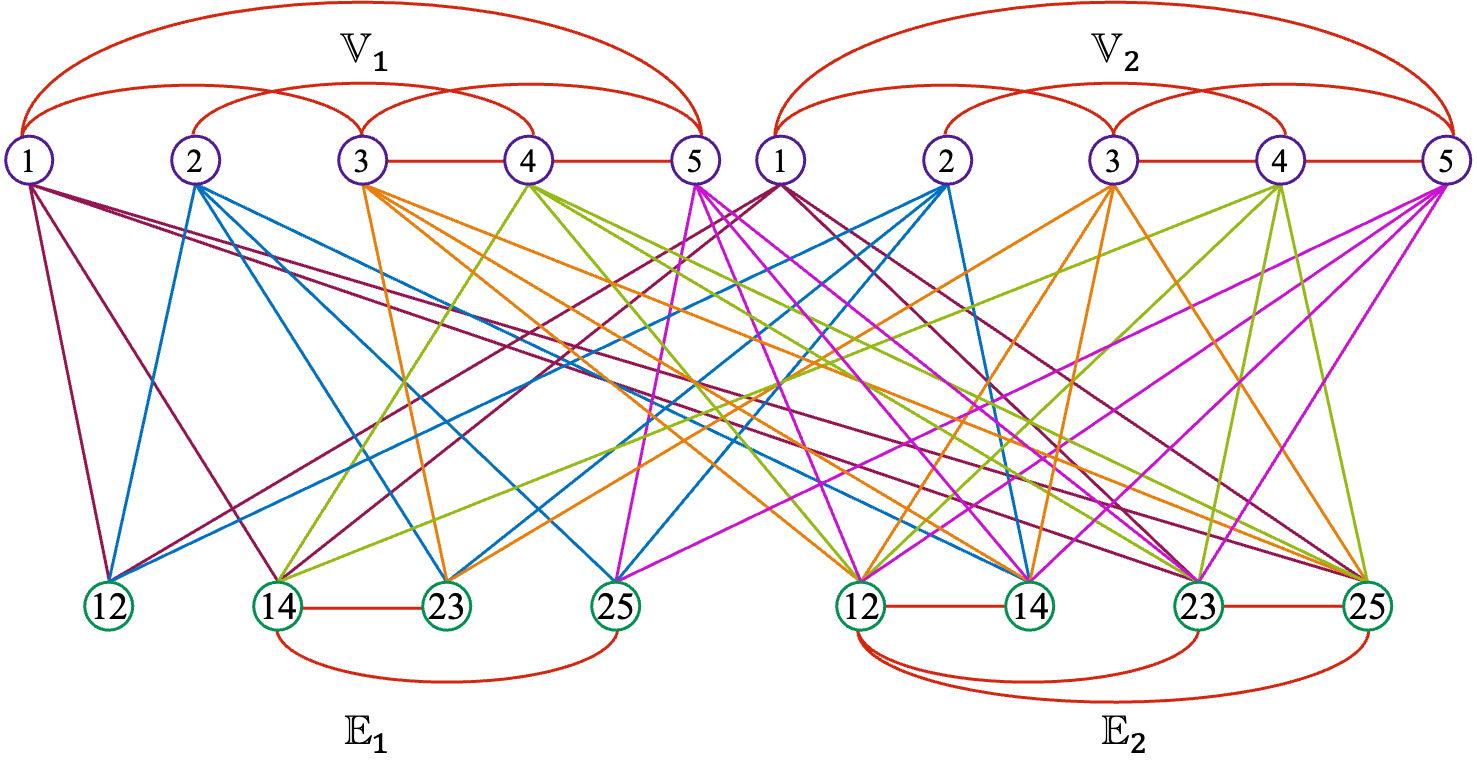}
			\caption{The transformation graph, $\mathbb{Q}_{22}^{-yz}$, where	$y=\{-,+\}$ and
				$z=\{+,-,+,-\}$.}
			\label{fig:Q22eg}
		\end{figure}
	\end{itemize}
	
	\section{$(r,s)$- Generalised Transformation Graphs $\mathbb{Q}_{rs}^{xy+}$ and $\mathbb{Q}_{rs}^{xy-}$}
	For a graph $\mathbb{Q}$, the transformation graph $\mathbb{Q}_{rs}^{xy+}$ has $x=\{x_1, x_2, ..., x_r\}$ and $y=\{y_1, y_2, ..., y_s\}$, where, $x_g, y_h \in \{+,-\}$, but $z=\{z_{11}, z_{12}, ..., z_{1s}, z_{21}, z_{22}, ..., z_{2s}, ..., z_{r1}, z_{r2}, ..., z_{rs}\}$ is such that $z_{gh}=+; \text{where, } 1\leq g \leq r$ and $1\leq h \leq s$.\\\\
	Similarly, the transformation graph $\mathbb{Q}_{rs}^{xy-}$ has $x=\{x_1, x_2, ..., x_r\}$ and $y=\{y_1, y_2, ..., y_s\}$, where, $x_g, y_h \in \{+,-\}$, but $z=\{z_{11}, z_{12}, ..., z_{1s}, z_{21}, z_{22}, ..., z_{2s}, ..., z_{r1}, z_{r2}, ..., z_{rs}\}$ is such that $z_{gh}=-; \text{where, } 1\leq g \leq r$ and $1\leq h \leq s$.\\
	
	Let $a(p)={(a_i)}_1^p$ and $b(q)={(b_j)}_1^q$ be two finite sequences such that $1\leq a_1 < a_2 < ... < a_p \leq r$ and  $1\leq b_1 < b_2 < ... < b_q \leq s$.\\ \\
	Define $x(a(p))=\{x_1, x_2, ... , x_r\}$, where
	\[
	x_g= 
	\begin{cases}
		+,& \text{if } g=a_i, 1\leq i \leq p;\\
		-,              & \text{otherwise}
	\end{cases}
	\]
	\\and, define $y(b(q))=\{y_1, y_2, ... , y_s\}$, where
	\[
	y_h= 
	\begin{cases}
		+,& \text{if } h=b_j, 1\leq j \leq q;\\
		-,              & \text{otherwise}
	\end{cases}
	\]
	\\Also, for $1\leq p \leq r$, define $x(p)=\{x_1, x_2, ... , x_r\}$, where
	\[
	x_g= 
	\begin{cases}
		+,& \text{if } 1\leq g \leq p;\\
		-,              & p < g \leq r
	\end{cases}
	\]
	\\ and for $ 1\leq q \leq s$, define $y(q)=\{y_1, y_2, ... , y_s\}$, where
	\[
	y_h= 
	\begin{cases}
		+,& \text{if } 1\leq h \leq q;\\
		-,              & q < h \leq s
	\end{cases}
	\]
	Similarly, define $x(\overline{a(p)})=\{x_1, x_2, ... , x_r\}$, where
	\[
	x_g= 
	\begin{cases}
		-,& \text{if } g=a_i, 1\leq i \leq p;\\
		+,              & \text{otherwise}
	\end{cases}
	\]
	\\and, define $y(\overline{b(q)})=\{y_1, y_2, ... , y_s\}$, where
	\[
	y_h= 
	\begin{cases}
		-,& \text{if } h=b_j, 1\leq j \leq q;\\
		+,              & \text{otherwise}
	\end{cases}
	\]
	\\Also, for $1\leq p \leq r$, define $x(\overline{p})=\{x_1, x_2, ... , x_r\}$, where
	\[
	x_g= 
	\begin{cases}
		-,& \text{if } 1\leq g \leq p;\\
		+,              & p < g \leq r
	\end{cases}
	\]
	\\ and for $ 1\leq q \leq s$, define $y(\overline{q})=\{y_1, y_2, ... , y_s\}$, where
	\[
	y_h= 
	\begin{cases}
		-,& \text{if } 1\leq h \leq q;\\
		+,              & q < h \leq s
	\end{cases}
	\]
 It is straight forward to see that
	\begin{enumerate}
		\item $\mathbb{Q}_{rs}^{x(a(p)) y(b(q)) +} \cong \mathbb{Q}_{rs}^{x(p) y(q) +}$ 
		\item $\mathbb{Q}_{rs}^{x(a(p)) y(b(q)) -} \cong \mathbb{Q}_{rs}^{x(p) y(q) -}$
		\item $\mathbb{Q}_{rs}^{x(\overline{a(p)}) y(\overline{b(q)}) +} \cong \mathbb{Q}_{rs}^{x(\overline{p}) y(\overline{q}) +}$ 
		\item $\mathbb{Q}_{rs}^{x(\overline{a(p)}) y(\overline{b(q)}) -} \cong \mathbb{Q}_{rs}^{x(\overline{p}) y(\overline{q}) -}$
	\end{enumerate}
	
	\subsection*{Special Cases:}
	\begin{enumerate}
		\item Consider the transformation graph, $\mathbb{Q}_{rs}^{x(p) y(q)+}$. If
		\begin{itemize}
			\item $p=r$ and $q=s$, then the transformation graph is denoted as $\mathbb{Q}_{rs}^{+++}$ 
			\item $p=r$ and $q=0$, then the transformation graph is denoted as$\mathbb{Q}_{rs}^{+-+}$ 
			\item  $p=0$ and $q=s$,then the transformation graph is denoted as$\mathbb{Q}_{rs}^{-++}$ 
			\item  $p=0$ and $q=0$,then the transformation graph is denoted as$\mathbb{Q}_{rs}^{--+}$ 	
		\end{itemize}
		\item Consider the transformation graph, $\mathbb{Q}_{rs}^{x(p) y(q)-}$, if
		\begin{itemize}
			\item $p=r$ and $q=s$, then the transformation graph is denoted as$\mathbb{Q}_{rs}^{++-}$ 
			\item $p=r$ and $q=0$, then the transformation graph is denoted as$\mathbb{Q}_{rs}^{+--}$ 
			\item  $p=0$ and $q=s$, then the transformation graph is denoted as$\mathbb{Q}_{rs}^{-+-}$ 
			\item  $p=0$ and $q=0$, then the transformation graph is denoted as$\mathbb{Q}_{rs}^{---}$ 	
		\end{itemize}
	\end{enumerate}
	
	\subsection{First Zagreb Index for $\mathbb{Q}_{rs}^{x(p) y(q)+}$}
	
	\begin{thm}
		Let $\mathbb{Q(V,E)}$ be a graph with $|\mathbb{V(Q)}|=\mathfrak{n}$ and $|\mathbb{E(Q)}|=\mathfrak{m}$. Then, for the transformation graph $\mathbb{Q}_{rs}^{x(p) y(q)+}$, $1\leq p \leq r; 1\leq q \leq s$,
		\begin{align*}
			\mathbb{M}_1\bigg(\mathbb{Q}_{rs}^{x(p) y(q)+}\bigg)=&\big\{p(s+1)^2 + (r-p)(s-1)^2 + 4q(r-1) + 2(q-s)(\mathfrak{m}+2r+1)\big\}\mathbb{M}_1(\mathbb{Q}) \\
			&+ 2s \mathbb{M}_2(\mathbb{Q}) + s F(\mathbb{Q}) + (r-p) \big\{\mathfrak{n}(\mathfrak{n}-1)^2+2\mathfrak{m}(\mathfrak{n}-1)(s-1)\big\}  \\&+ \mathfrak{m}q(2r-2)^2+ \mathfrak{m}(s-q)(\mathfrak{m}+2r+1)^2
		\end{align*}
	\end{thm}
	\begin{proof}
		Consider the graph $\mathbb{Q}_{rs}^{x(p) y(q)+}$.\\
		It's vertex set $\mathbb{V}\bigg(\mathbb{Q}_{rs}^{x(p) y(q)+}\bigg)= \bigg(\displaystyle\bigcup_{g=1}^{r} \mathbb{V}_g \bigg)\bigcup \bigg(\displaystyle\bigcup_{h=1}^{s} \mathbb{E}_h \bigg)$.\\
		For $p; 1\leq p \leq r$, let 
		\[
		x_g= 
		\begin{cases}
			+,& \text{if } 1\leq g \leq p;\\
			-,              & p+1\leq g \leq r
		\end{cases}
		\]
		and for $q; 1\leq q \leq s$, let 
		\[
		y_h= 
		\begin{cases}
			+,& \text{if } 1\leq h \leq q;\\
			-,              & q+1\leq h \leq s
		\end{cases}
		\]
		\\
		Clearly, the elements in $z$ are all $'+'$.\\
		Hence, $x$ has $p$ number of $'+'$ and $r-p$ number of $'-'$. Similarly, $y$ has $q$ number of $'+'$ and $s-q$ number of $'-'$ and $z$ has $rs$ elements, which are all $'+'$.\\\\
		Under the given conditions $\mathbb{V}(\mathbb{Q}_{rs}^{x(p) y(q)+})$ have two types of vertices, say, 
		\begin{enumerate}[label=(\roman*)]
			\item $\alpha \in \mathbb{V}(\mathbb{Q}_{rs}^{x(p) y(q)+}) \cap \mathbb{V}_g$, where $1 \leq g \leq r$.
			\item $\gamma=\beta \delta \in \mathbb{V}(\mathbb{Q}_{rs}^{x(p) y(q)+}) \cap \mathbb{E}_h$, where $1 \leq h \leq s$.\\
		\end{enumerate}
		The degree of each vertex in $\mathbb{Q}_{rs}^{x(p) y(q)+}$ can be expressed as given.\\\\
		For a vertex $\alpha \in \mathbb{V}(\mathbb{Q}_{rs}^{x(p) y(q)+}) \cap \mathbb{V}_g$, where $1 \leq g \leq p$,
		\begin{align} 
			d_{(\mathbb{Q}_{rs}^{x(p) y(q)+}) \cap \mathbb{V}_g}(\alpha)=(s+1) d_{\mathbb{Q}}(\alpha).
		\end{align}
		For a vertex $\alpha \in \mathbb{V}(\mathbb{Q}_{rs}^{x(p) y(q)+}) \cap \mathbb{V}_g$, where $(p+1) \leq g \leq r$, \begin{align}
			d_{(\mathbb{Q}_{rs}^{x(p) y(q)+}) \cap \mathbb{V}_g}(\alpha)=(s-1) d_{\mathbb{Q}}(\alpha)+\mathfrak{n}-1.
		\end{align}
		For a vertex $\gamma = \beta \delta \in \mathbb{V}(\mathbb{Q}_{rs}^{x(p) y(q)+}) \cap \mathbb{E}_h$, where $1 \leq h \leq q$, \begin{align} 
			d_{(\mathbb{Q}_{rs}^{x(p) y(q)+}) \cap \mathbb{E}_h}(\gamma)= d_{\mathbb{Q}}(\beta)+d_{\mathbb{Q}}(\delta)+2(r-1).
		\end{align}
		For a vertex $\gamma=\beta \delta \in \mathbb{V}(\mathbb{Q}_{rs}^{x(p) y(q)+}) \cap \mathbb{E}_h$, where $(q+1) \leq h \leq s$, \begin{align} 
			d_{(\mathbb{Q}_{rs}^{x(p) y(q)+}) \cap \mathbb{E}_h}(\gamma)=\mathfrak{m}+2r+1- \big(d_{\mathbb{Q}}(\beta)+d_{\mathbb{Q}}(\delta)\big).
		\end{align}
		Applying these results, the first Zagreb index for $\mathbb{Q}_{rs}^{x(p) y(q)+}$ is computed as follows.\\\\
		From Equation 1, 
		\begin{align*}
			\mathbb{M}_1\mathbb{(Q)}=\sum_{\alpha \in \mathbb{V(Q)}}\big(d_{\mathbb{Q}}(\alpha)\big)^2.
		\end{align*}
		Therefore,
		\begin{align*}
			\mathbb{M}_1(\mathbb{Q}_{rs}^{x(p) y(q)+})=&\sum_{\alpha \in \mathbb{V}(\mathbb{Q}_{rs}^{x(p) y(q)+})}\big(d_{\mathbb{Q}_{rs}^{x(p) y(q)+}}(\alpha)\big)^2\\
			=&\sum_{\alpha \in \mathbb{V}(\mathbb{Q}_{rs}^{x(p) y(q)+}) \cap \mathbb{V}_g} \big(d_{\mathbb{Q}_{rs}^{x(p) y(q)+}}(\alpha)\big)^2 + \sum_{\gamma \in \mathbb{V}(\mathbb{Q}_{rs}^{x(p) y(q)+}) \cap \mathbb{E}_h} \big(d_{\mathbb{Q}_{rs}^{x(p) y(q)+}}\big)^2\\\\
			=& \hspace{0.2cm} p  \sum_{\alpha \in \mathbb{V(Q)}} \bigg((s+1) d_{\mathbb{Q}}(\alpha)\bigg)^2 +
			(r-p)  \sum_{\alpha \in \mathbb{V(Q)}} \bigg((s-1) d_{\mathbb{Q}}(\alpha)+\mathfrak{n}-1 \bigg)^2 
			\\&+ q  \sum_{\beta \delta \in \mathbb{E(Q)}} \bigg(d_{\mathbb{Q}}(\beta)+d_{\mathbb{Q}}(\delta)+2(r-1)\bigg)^2
			\\&+(s-q)  \sum_{\beta \delta \mathbb{E(Q)}} \bigg(\mathfrak{m}+2r+1- (d_{\mathbb{Q}}(\beta)+d_{\mathbb{Q}}(\delta))\bigg)^2\\\\
			=& \hspace{0.2cm} p (s+1)^2 \hspace{0.2cm} \mathbb{M}_1\mathbb{(Q)} +
			(r-p)  \big( (s-1)^2 \hspace{0.2cm} \mathbb{M}_1\mathbb{(Q)} + \mathfrak{n}(\mathfrak{n}-1)^2 + 2\mathfrak{m}(\mathfrak{n}-1)(s-1) \big)
			\\&+ q  \big( 2 \hspace{0.2cm} \mathbb{M}_2\mathbb{(Q)} + \hspace{0.2cm} \mathbb{F}\mathbb{(Q)}+ \mathfrak{m}(2r-2)^2+ 4(r-1) \hspace{0.2cm} \mathbb{M}_1\mathbb{(Q)}\big)
			\\&+(s-q)  \big( 2 \hspace{0.2cm} \mathbb{M}_2\mathbb{(Q)} + \hspace{0.2cm} \mathbb{F}\mathbb{(Q)}+\mathfrak{m}(\mathfrak{m}+2r+1)^2-2(\mathfrak{m}+2r+1) \hspace{0.2cm} \mathbb{M}_1\mathbb{(Q)} \big)\\
			\\=&\hspace{0.2cm} \big\{p(s+1)^2 + (r-p)(s-1)^2 + 4q(r-1) + 2(q-s)(\mathfrak{m}+2r+1)\big\}\mathbb{M}_1(\mathbb{Q}) \\
			&+ 2s \mathbb{M}_2(\mathbb{Q}) + s F(\mathbb{Q}) + (r-p) \big\{\mathfrak{n}(\mathfrak{n}-1)^2+2\mathfrak{m}(\mathfrak{n}-1)(s-1)\big\} + \mathfrak{m}q(2r-2)^2 \\&+ \mathfrak{m}(s-q)(\mathfrak{m}+2r+1)^2	
		\end{align*}
	\end{proof}
	
	\begin{corollary}
			Consider the graph $\mathbb{Q(V,E)}$, such that $|\mathbb{V(Q)}|=\mathfrak{n}$ and $|\mathbb{E(Q)}|=\mathfrak{m}$, then
		\begin{align*}
			\mathbb{M}_1(\mathbb{Q}_{rs}^{+++})=&\big(r(s^2+6s+1)-4s\big)\mathbb{M}_1(\mathbb{Q}) + 2s \mathbb{M}_2(\mathbb{Q}) + s F(\mathbb{Q}) + 4\mathfrak{m}s(r-1)^2 \\
			\mathbb{M}_1(\mathbb{Q}_{rs}^{+-+})=&\big(r(s+1)^2-2s(\mathfrak{m}+2r+1)\big)\mathbb{M}_1(\mathbb{Q}) + 2s \mathbb{M}_2(\mathbb{Q}) + s F(\mathbb{Q}) + \mathfrak{m}s(\mathfrak{m}+2r+1)^2 \\
			\mathbb{M}_1(\mathbb{Q}_{rs}^{-++})=&\big(r(s-1)^2-2s(\mathfrak{m}+2r+1)\big)\mathbb{M}_1(\mathbb{Q}) + 2s \mathbb{M}_2(\mathbb{Q}) + s F(\mathbb{Q}) + \mathfrak{n}r(\mathfrak{n}-1)^2 \\& + 2\mathfrak{m} r(\mathfrak{n}-1)(s-1)+ 4\mathfrak{m} s(r-1)^2
			\\
			\mathbb{M}_1(\mathbb{Q}_{rs}^{--+})=&\big(r(s-1)^2-4sr-4s\big)\mathbb{M}_1(\mathbb{Q}) + 2s \mathbb{M}_2(\mathbb{Q}) + s F(\mathbb{Q}) + \mathfrak{n}r(\mathfrak{n}-1)^2 \\& + 2\mathfrak{m} r(\mathfrak{n}-1)(s-1)+\mathfrak{m} s(\mathfrak{m} +2r+1)^2 
		\end{align*}
	\end{corollary}	
	
	\subsection{First Zagreb Index for $\mathbb{Q}_{rs}^{x(p) y(q)-}$}
	
	\begin{thm}
		Let $\mathbb{Q(V,E)}$ be a graph with $|\mathbb{V(Q)}|=\mathfrak{n}$ and $|\mathbb{E(Q)}|=\mathfrak{m}$. Then, for the transformation graph $\mathbb{Q}_{rs}^{x(p) y(q)-}$, $1\leq p \leq r; 1\leq q \leq s$,
		\begin{align*}
			\mathbb{M}_1\bigg(\mathbb{Q}_{rs}^{x(p) y(q)-}\bigg)=&\big\{p(1-s)^2 + (r-p)(s+1)^2 + 2q(\mathfrak{n}r-2r-2) \\&+ 2(q-s)(\mathfrak{m}+r(a-2)+1)\big\} \mathbb{M}_1(\mathbb{Q}) + 2s \mathbb{M}_2(\mathbb{Q}) + s F(\mathbb{Q}) \\&+ (r-p) \big\{\mathfrak{n}(\mathfrak{n}+sb-1)^2-4\mathfrak{m}(\mathfrak{n}+s\mathfrak{m}-1)(s+1)\big\} \\
			&+ps\mathfrak{m}^2(\mathfrak{n}s+4(1-s))+ \mathfrak{m}q(\mathfrak{n}r-2r-2)^2 + \mathfrak{m}(s-q)(\mathfrak{m}+r(\mathfrak{n}-2)+1)^2
		\end{align*}
	\end{thm}
	\begin{proof}
		Consider the graph $\mathbb{Q}_{rs}^{x(p) y(q)-}$. 
		It's vertex set $\mathbb{V}\bigg(\mathbb{Q}_{rs}^{x(p) y(q)-}\bigg)= \bigg(\displaystyle\bigcup_{g=1}^{r} \mathbb{V}_g \bigg)\bigcup \bigg(\displaystyle\bigcup_{h=1}^{s} \mathbb{E}_h \bigg)$.\\
		For $p; 1\leq p \leq r$ and $x=\{x_1, x_2, ..., x_r\}$,
		\[
		x_g= 
		\begin{cases}
			+,& \text{if } 1\leq g \leq p;\\
			-,              & p+1\leq h \leq r
		\end{cases}
		\]
		Also, for $q; 1\leq q \leq s$ and $y=\{y_1, y_2, ..., y_s\}$, 
		\[
		y_h= 
		\begin{cases}
			+,& \text{if } 1\leq h \leq q;\\
			-,              & q+1\leq h \leq s
		\end{cases}
		\]
		Clearly, the elements in $z$ are all $'-'$.
		Hence, $x$ has $p$ number of $'+'$ and $r-p$ number of $'-'$. Similarly, $y$ has $q$ number of $'+'$ and $s-q$ number of $'-'$ and $z$ has $rs$ elements, which are all $'-'$.\\\\
		Under the given conditions, $\mathbb{V}(\mathbb{Q}_{rs}^{x(p) y(q)-})$ have two types of vertices, say, 
		\begin{enumerate}[label=(\roman*)]
			\item $\alpha \in \mathbb{V}(\mathbb{Q}_{rs}^{x(p) y(q)-}) \cap \mathbb{V}_g$, where $1 \leq g \leq r$.
			\item $\gamma=\beta \delta \in \mathbb{V}(\mathbb{Q}_{rs}^{x(p) y(q)-}) \cap \mathbb{E}_h$, where $1 \leq h \leq s$.\\
		\end{enumerate}
		The degree of each vertex in $\mathbb{Q}_{rs}^{x(p) y(q)-}$ can be expressed as given.\\\\
		For a vertex $\alpha \in \mathbb{V}(\mathbb{Q}_{rs}^{x(p) y(q)-}) \cap \mathbb{V}_g$, where $1 \leq g \leq p$,
		\begin{align} 
			d_{(\mathbb{Q}_{rs}^{x(p) y(q)-}) \cap \mathbb{V}_g}(\alpha)= s \mathfrak{m} + (1-s)d_{\mathbb{Q}}(\alpha).
		\end{align}
		For a vertex $\alpha \in \mathbb{V}(\mathbb{Q}_{rs}^{x(p) y(q)-}) \cap \mathbb{V}_g$, where $(p+1) \leq g \leq r$, \begin{align}
			d_{(\mathbb{Q}_{rs}^{x(p) y(q)-}) \cap \mathbb{V}_g}(\alpha)=\mathfrak{n}+s\mathfrak{m}-1-(s+1) d_{\mathbb{Q}}(\alpha).
		\end{align}
		For a vertex $\gamma = \beta \delta \in \mathbb{V}(\mathbb{Q}_{rs}^{x(p) y(q)-}) \cap \mathbb{E}_h$, where $1 \leq h \leq q$, \begin{align} 
			d_{(\mathbb{Q}_{rs}^{x(p) y(q)-}) \cap \mathbb{E}_h}(\gamma)= d_{\mathbb{Q}}(\beta)+d_{\mathbb{Q}}(\delta)+r(\mathfrak{n}-2)-2.
		\end{align}
		For a vertex $\gamma=\beta \delta \in \mathbb{V}(\mathbb{Q}_{rs}^{x(p) y(q)-}) \cap \mathbb{E}_h$, where $(q+1) \leq h \leq s$, 
		\begin{align} 
			d_{(\mathbb{Q}_{rs}^{x(p) y(q)-}) \cap \mathbb{E}_h}(\gamma)=\mathfrak{m}+r(\mathfrak{n}-2)+1- \big(d_{\mathbb{Q}}(\beta)+d_{\mathbb{Q}}(\delta)\big).
		\end{align}
		Applying these results, the first Zagreb index for $\mathbb{Q}_{rs}^{x(p) y(q)-}$ is computed as in Theorem 1.
		\begin{align*}
			\mathbb{M}_1(&\mathbb{Q}_{rs}^{x(p) y(q)-})=\sum_{\alpha \in \mathbb{V}(\mathbb{Q}_{rs}^{x(p) y(q)-})}\big(d_{\mathbb{Q}_{rs}^{x(p) y(q)-}}(\alpha)\big)^2\\
			=&\sum_{\alpha \in \mathbb{V}(\mathbb{Q}_{rs}^{x(p) y(q)-}) \cap \mathbb{V}_g} \big(d_{\mathbb{Q}_{rs}^{x(p) y(q)-}}(\alpha)\big)^2 + \sum_{\gamma \in \mathbb{V}(\mathbb{Q}_{rs}^{x(p) y(q)-}) \cap \mathbb{E}_h} \big(d_{\mathbb{Q}_{rs}^{x(p) y(q)-}}\big)^2\\\\
			=& \hspace{0.2cm} p  \sum_{\alpha \in \mathbb{V(Q)}} \bigg(s\mathfrak{m}+(1-s)d_{\mathbb{Q}}(\alpha)\bigg)^2 +
			(r-p)  \sum_{\alpha \in \mathbb{V(Q)}} \bigg(\mathfrak{n}+s\mathfrak{m}-1-(s+1) d_{\mathbb{Q}}(\alpha) \bigg)^2 
			\\&+ q  \sum_{\beta \delta \in \mathbb{E(Q)}} \bigg(d_{\mathbb{Q}}(\beta)+d_{\mathbb{Q}}(\delta)+r(\mathfrak{n}-2)-2\bigg)^2
			\\&+(s-q)  \sum_{\beta \delta \mathbb{E(Q)}} \bigg(\mathfrak{m}+r(\mathfrak{n}-2)+1- \big(d_{\mathbb{Q}}(\beta)+d_{\mathbb{Q}}(\delta)\big)\bigg)^2\\\\
			=&\big\{p(1-s)^2 + (r-p)(s+1)^2 + 2q(\mathfrak{n}r-2r-2) + 2(q-s)(\mathfrak{m}+r(\mathfrak{n}-2)+1)\big\} \mathbb{M}_1(\mathbb{Q}) \\
			&+ 2s \mathbb{M}_2(\mathbb{Q}) + s F(\mathbb{Q}) + (r-p) \big\{\mathfrak{n}(\mathfrak{n}+s\mathfrak{m}-1)^2-4\mathfrak{m}(\mathfrak{n}+s\mathfrak{m}-1)(s+1)\big\} \\
			&+ps\mathfrak{m}^2(\mathfrak{n}s+4(1-s))+ \mathfrak{m}q(\mathfrak{n}r-2r-2)^2 + \mathfrak{m}(s-q)(\mathfrak{m}+r(\mathfrak{n}-2)+1)^2
		\end{align*}
	\end{proof}
	
	\begin{corollary}
	Consider the graph $\mathbb{Q(V,E)}$, such that $|\mathbb{V(Q)}|=\mathfrak{n}$ and $|\mathbb{E(Q)}|=\mathfrak{m}$, then
		\begin{align*}
			\mathbb{M}_1(\mathbb{Q}_{rs}^{++-})=& \big(r(1-s)^2+2s(\mathfrak{n}r-2r-2))\big)\mathbb{M}_1(\mathbb{Q}) + 2s \mathbb{M}_2(\mathbb{Q}) + s F(\mathbb{Q}) + \mathfrak{n}r\mathfrak{m}^2s^2\\&+ 4rs\mathfrak{m}^2(1-s)+s\mathfrak{m}(\mathfrak{n}r-2r-2)^2
			\\
			\mathbb{M}_1(\mathbb{Q}_{rs}^{+--})=&\big(r(1-s)^2-2s(\mathfrak{m}+r(\mathfrak{n}-2)+1)\big)\mathbb{M}_1(\mathbb{Q}) + 2s \mathbb{M}_2(\mathbb{Q}) + s F(\mathbb{Q}) + \mathfrak{n}r\mathfrak{m}^2s^2\\&+ 4rs\mathfrak{m}^2(1-s)+s\mathfrak{m}(\mathfrak{m}+r(\mathfrak{n}-2)+1)^2 
		\end{align*}
		\begin{align*}
			\mathbb{M}_1(\mathbb{Q}_{rs}^{-+-})=&\big(r(s+1)^2+2s(\mathfrak{n}r-2r-2)\big)\mathbb{M}_1(\mathbb{Q}) + 2s \mathbb{M}_2(\mathbb{Q}) + s F(\mathbb{Q}) + \mathfrak{n}r(\mathfrak{n}+s\mathfrak{m}-1)^2 \\&- 4\mathfrak{m}r(s+1)(\mathfrak{n}+sm-1)+\mathfrak{m}s(\mathfrak{n}r-2r-2)^2
			\\
			\mathbb{M}_1(\mathbb{Q}_{rs}^{---})=&\big(r(s+1)^2-2s(\mathfrak{m}+r(\mathfrak{n}-2)+1)\big)\mathbb{M}_1(\mathbb{Q}) + 2s \mathbb{M}_2(\mathbb{Q}) + s F(\mathbb{Q}) + \mathfrak{n}r(\mathfrak{n}+s\mathfrak{m}-1)^2 \\& -4 \mathfrak{m}r(\mathfrak{n}+\mathfrak{m}s-1)(s+1)+\mathfrak{m}s(\mathfrak{m}+r(\mathfrak{n}-2)+1)^2 
		\end{align*}
	\end{corollary}	
	
	\section{Conclusion}
	In this work, the existing concept of generalised transformation graphs is elevated to next level by taking union of $r$ copies of vertex set and $s$ copies of edge set of a graph $\mathbb{Q}$ as the vertex set of the transformed graph givng rise to $(r,s)$- generalised transformation graphs. Under this concept, rather than the eight transformations, we can generate infinite number of transformations for a given graph. Further, the first Zagreb index for few type of $(r,s)$- generalised transformation graphs were determined and analysed. 
	
		\section*{Conflict of Interest}
	The authors declare no conflict of interest.
	
		\section*{Declaration}
	There is no data availability.

	\bibliographystyle{plain}
		
\end{document}